\pdfoutput=1  
\documentclass[]{article}  
\usepackage{url,float}
\usepackage{graphicx}
\usepackage{amsmath}
\usepackage{amsfonts}
\usepackage{amssymb}
\usepackage{latexsym}

\newcommand{\hide}[1]{}

\newcommand{\ABox}{
\raisebox{3pt}{\framebox[6pt]{\rule{6pt}{0pt}}}
}
\newenvironment{proof}{{\bf Proof:}}{\hfill\ABox}

\newtheorem{theorem}{{\bf Theorem}}
\newtheorem{corollary}{Corollary}
\newtheorem{lemma}{Lemma}
\newtheorem{proposition}{Proposition}

\newcommand{\lemlab}[1]{\label{lemma:#1}}
\newcommand{\thmlab}[1]{\label{thm:#1}}
\newcommand{\corlab}[1]{\label{cor:#1}}

\newcommand{\figlab}[1]{\label{fig:#1}}
\newcommand{\seclab}[1]{\label{sec:#1}}

\newcommand{\lemref}[1]{\ref{lemma:#1}}
\newcommand{\thmref}[1]{\ref{thm:#1}}
\newcommand{\corref}[1]{\ref{cor:#1}}

\newcommand{\secref}[1]{\ref{sec:#1}}
\newcommand{\figref}[1]{\ref{fig:#1}}

{\makeatletter
 \gdef\xxxmark{%
   \expandafter\ifx\csname @mpargs\endcsname\relax 
     \expandafter\ifx\csname @captype\endcsname\relax 
       \marginpar{xxx}
     \else
       xxx 
     \fi
   \else
     xxx 
   \fi}
 \gdef\xxx{\@ifnextchar[\xxx@lab\xxx@nolab}
 \long\gdef\xxx@lab[#1]#2{{\bf [\xxxmark #2 ---{\sc #1}]}}
 \long\gdef\xxx@nolab#1{{\bf [\xxxmark #1]}}
 \gdef\turnoffxxx{\long\gdef\xxx@lab[##1]##2{}\long\gdef\xxx@nolab##1{}}%
}

\def\o{{\omega}}
\def\a{{\alpha}}
\def\b{{\beta}}

\def\R{{\mathbb{R}}}
\def\Q{{\mathbb{Q}}}

\title{%
Common Edge-Unzippings for Tetrahedra
} 

\author{%
Joseph O'Rourke%
    \thanks{Department of Computer Science, Smith College, Northampton, MA
      01063, USA.
      \protect\url{orourke@cs.smith.edu}.}
}

\begin{document}
\maketitle

\begin{abstract}
It is shown that there are examples of
distinct polyhedra, each with a Hamiltonian path of edges, which 
when cut, unfolds the surfaces to a common net.
In particular, it is established for infinite classes of triples of tetrahedra.
\end{abstract}

\section{Introduction}
\seclab{Introduction}
The limited focus of this note is to establish that there are
an infinite collection of
``edge-unfolding zipper pairs'' of convex polyhedra.
A  \emph{net} for a convex polyhedron $P$ is an unfolding
of its surface to a planar simple
(nonoverlapping) polygon, obtained by cutting
a spanning tree of its edges (i.e., of its 1-skeleton);
see~\cite[Sec.~22.1]{do-gfalop-07}.
Shephard explored in the 1970's the special case where the spanning tree is
a Hamiltonian path of edges on $P$~\cite{s-cpcn-75}.
Such \emph{Hamiltonian unfoldings} were futher studied
in~\cite{ddlo-efupp-02} (see~\cite[Fig.~25.59 ]{do-gfalop-07}),
and most recently in~\cite{lddss-zupc-10}, 
where the natural term \emph{zipper unfolding} was introduced.
Define two polyhedra to be an \emph{edge-unfolding zipper pair}
if they have a zipper unfolding to a common net.
Here we emphasize edge-unfolding, as opposed to an arbitrary
zipper path that may cut through the interior of faces,
which are easier to identify. 
Thus we are considering a special case of more general 
\emph{net pairs}:
pairs of polyhedra that may be cut open to a common net.

In general there is little understanding of which polyhedra form net
pairs under any definition.
See, for example, Open Problem~25.6 in~\cite{do-gfalop-07},
and~\cite{o-fzupp-10} for an exploration of Platonic solids.
Here we establish that there are infinite classes of convex
polyhedra that form edge-unfolding zipper pairs.
Thus one of these polyhedra can be cut open along an edge zipper path,
and rezipped to form a different polyhedron with the same property.
In particular, we prove this theorem:

\begin{theorem}
Every equilateral convex hexagon, with 
each angle in the range $(\pi/3,\pi)$,
and each pair of angles
linearly independent over $\Q$,
is the common edge-unzipping of three incongruent tetrahedra.
\thmlab{HexTheorem}
\end{theorem}

The angle restrictions in the statement of the theorem are,
in some sense, incidental, included to match the proof techniques.
The result is quite narrow, and although certainly generalizations
hold, there are impediments to proving them formally.  This issue will be
discussed in Section~\secref{Discussion}.

\section{Example}
\seclab{Example}
An example is shown in Figure~\figref{Tetrahedra3}.
The equilateral hexagon is folded via a \emph{perimeter halving}
folding~\cite[Sec.~25.1.2]{do-gfalop-07}:
half the perimeter determined by opposite vertices
is glued (``zipped'') to the other half, matching the other four
vertices
in two pairs.  Because the hexagon is equilateral, the corresponding
edge lengths match.
\begin{figure}[htbp]
\centering
\includegraphics[width=\linewidth]{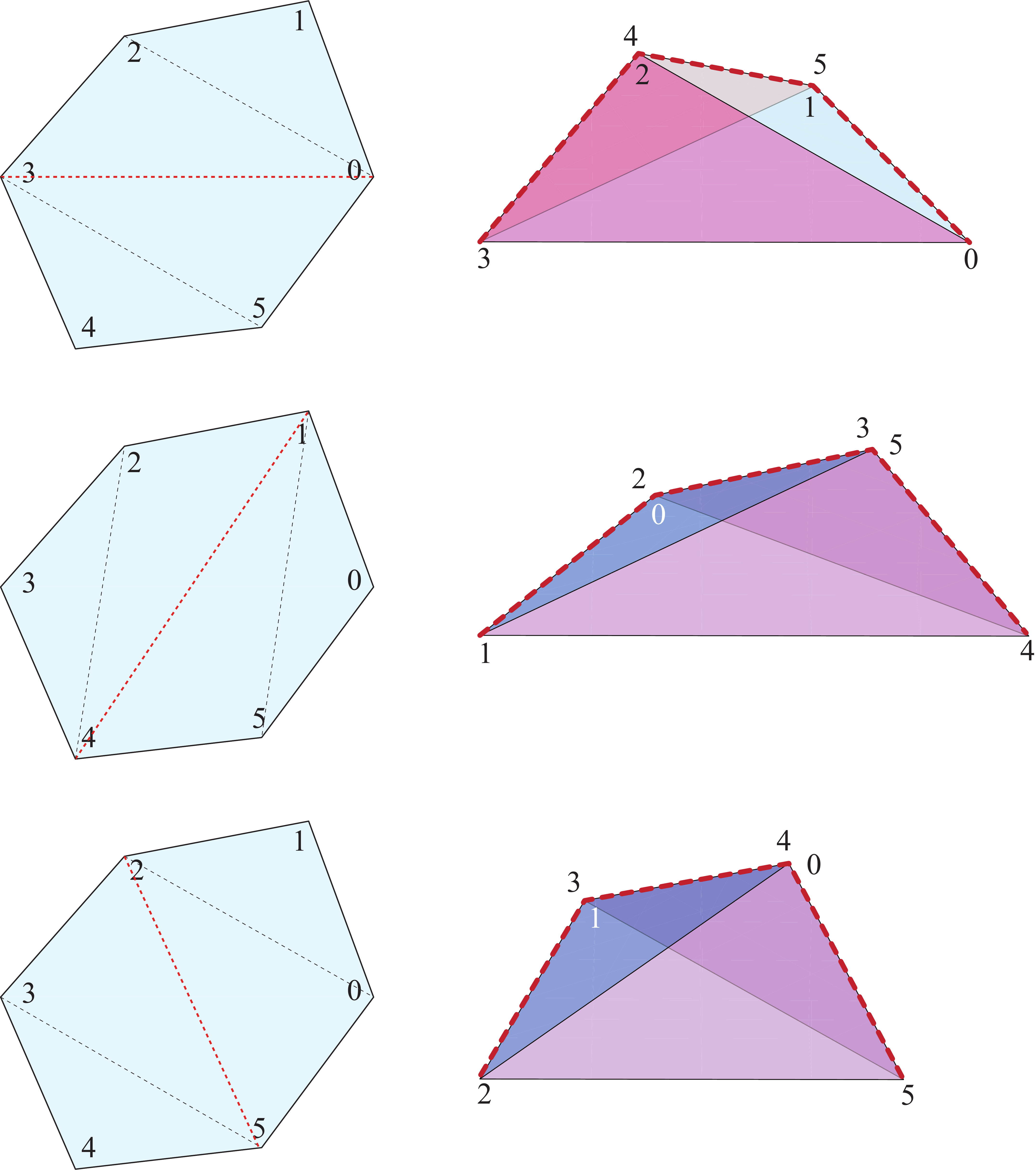}
\caption{Three tetrahedra that edge-unzip
(via the dashed Hamiltonian path) to
a common equilateral hexagon.}
\figlab{Tetrahedra3}
\end{figure}

The resulting shape is a convex polyhedron by Alexandrov's theorem
(\cite[Sec.~23.3]{do-gfalop-07}),
and a tetrahedron because it has four points at which the curvature is
nonzero.
Choosing to halve the perimeter at each of the three pairs of opposite
vertices leads to the three tetrahedra illustrated.
Note that each has three unit-length edges that form a Hamiltonian
zipper path, as claimed in the theorem.

The main challenge is to show that the edges of the hexagon glued
together in
pairs become edges of the polyhedron, and so the unzipping is an edge-unzipping.

\section{Distinct Tetrahedra}
\seclab{Distinct}
First we show that the linear independence condition in
Theorem~\thmref{HexTheorem}
ensures the tetrahedra will be distinct.
We define two angles $x$ and $y$ to be \emph{linearly independent over $\Q$}
if there are no rationals $a$ and $b$ such that $y = a \pi + b x$.
We will abbreviate this condition to ``linear independence'' below.

We will show that the curvatures 
($2\pi$ minus the incident face angle)
at the vertices 
of the tetrahedra
differ
at one vertex or more, which ensures that the tetrahedra are not
congruent
to one another.
Let the vertices of the hexagon be $v_0, \ldots, v_5$, with
vertex angles $\a_0, \ldots, \a_5$.

We identify the three tetrahedra by their halving diagonals,
and name them
$T_{03}$, $T_{14}$, $T_{25}$.
We name the curvatures at the four vertices of the tetrahedra $\o_1$, $\o_2$, $\o_3$, $\o_4$.
These curvatures for the three tetrahedra 
(see Figure~\figref{Tetrahedra3})
are as
follows:

\begin{center}
\begin{tabular}{|c||c|c|c|c|}
\hline \hline
Tetrahedron & \multicolumn{4}{|c|}{Vertex Curvatures}
\\ \cline{2-5}
halving diagonal & $\o_1$ & $\o_2$ & $\o_3$ & $\o_4$
\\ \hline \hline
$T_{03} \;:\; v_0 v_3$ 
   & $\frac{1}{2} (2 \pi - \a_0)$ 
   & $\frac{1}{2} (2 \pi -\a_3)$ 
   & $2 \pi - (\a_1 + \a_5)$
   &  $2 \pi - (\a_2 + \a_4)$
\\ \hline

$T_{14} \;:\; v_1 v_4$ 
   & $\frac{1}{2} (2 \pi - \a_1)$ 
   & $\frac{1}{2} (2 \pi -\a_4)$ 
   & $2 \pi - (\a_0 + \a_2)$
   &  $2 \pi - (\a_3 + \a_5)$
\\ \hline

$T_{25} \;:\;  v_2 v_5$ 
   & $\frac{1}{2} (2 \pi - \a_2)$ 
   & $\frac{1}{2} (2 \pi -\a_5)$ 
   & $2 \pi - (\a_0 + \a_1)$
   &  $2 \pi - (\a_3 + \a_4)$

\\ \hline \hline
\end{tabular}
\end{center}

Now we explore under what conditions could the four curvatures of
$T_{03}$ be identical to the four curvatures of $T_{14}$.
(There is no need to explore other possibilies, as they are all
equivalent to this situation by relabeling the hexagons.)
Let us label the curvatures of $T_{03}$ without primes, and those of $T_{14}$ with primes.
First note that we cannot have $\o_1 = \o'_1$ or $\o_1 = \o'_2$,
for then two angles must be equal: $\a_0 = a_1$ or $\a_0 = \a_4$
respectively.
And equal angles are not linearly independent.

So we are left with these possibilities: $\o_1 = \o'_3$, or $\o_1 =
\o'_4$.
The first leads to the relationship $\a_2 = \pi - \frac{1}{2} \a_0$,
a violation of linear independence.
The second possibility, $\o_1 = \o'_4$, requires further analysis.

It is easy to eliminate all but these two possibilities, which map
indices $(1,2,3,4)$ to either $(4',3',1',2')$ or $(4',3',2',1')$:
$$ 
\o_1 = \o'_4 \;,\; 
\o_2 = \o'_3  \;,\;
\o_3 = \o'_1  \;,\;
\o_4 = \o'_2
$$
and
$$ 
\o_1 = \o'_4 \;,\; 
\o_2 = \o'_3  \;,\;
\o_3 = \o'_2  \;,\;
\o_4 = \o'_1
$$
Explicit calculation shows that the first set implies that 
$\a_4 = 2 \pi - 2 \a_2$,
and the second implies that
$\a_1 = \frac{3}{4} \pi - \a_0$.

Thus we have reached the desired conclusion:
\begin{lemma}
Linear independence over $\Q$ of pairs of angles
of the hexagon (as stated in Theorem~\thmref{HexTheorem})
implies that the three tetrahedra have distinct vertex
curvatures, and therefore are incongruent polyhedra.
\lemlab{Distinct}
\end{lemma}

We have phrased the condition as linear independence of pairs of
angles
for simplicity, but in fact a collection of linear relationships must
hold
for the four curvatures to be equal.
So the restriction could be phrased more narrowly.
Note also we have not used in the proof of Lemma~\lemref{Distinct}
the restriction that the angles all be ``fat'', $\a_i > \pi/3$.
This will be used only in Lemma~\lemref{Edges} below.

It would be equally possible to rely on edge lengths rather than
angles to force distinctness of the tetrahedra.  For example, we could demand that
no two diagonals of the hexagon have the same length
(but that would leave further work).
Another alternative is to avoid angle restrictions, permitting all equilateral hexagons, but
only conclude that ``generally'' the three tetrahedra are distinct.
Obviously a regular hexagon leads to three identical tetrahedra.

The form of Theorem~\thmref{HexTheorem} as stated has the advantage of
easily implying that an uncountable number of hexagons satisfy its conditions
(see Corollary~\corref{Continuum}).

\section{Shortest Paths are Edges}
\seclab{ShortestPaths}
Now we know that the hexagons fold to three distinct tetrahedra.
So there is a zipper path on each tetrahedron that unfolds it
to that common hexagon.
All the remaining work is to show that the zipper path is an
edge path---composed of polyhedron edges.  
This seems to be less straightforward than one might
expect,
largely because there are not many tools available beyond Alexandrov's
existence theorem.

It is easy to see that every edge of a polyhedron $P$ is a shortest
path
between its endpoints.
The reverse is far from true in general, but it holds for tetrahedra:

\begin{lemma}
Each shortest path between vertices of a tetrahedron is realized by
an edge
of the tetrahedron.
\lemlab{ShortestPaths}
\end{lemma}
\begin{proof}
Note that there are $\binom{4}{2} = 6$ shortest paths between the
four
vertices, and six edges in a tetrahedron, so the combinatorics are
correct.

Suppose a path $\rho=xy$ is a shortest path between vertices $x$ and
$y$ of a tetrahedron $T$, but not an edge of $T$.  Because each pair
of vertices
of a tetrahedron is connected by an edge, there is an edge $e=xy$ of
$T$.
Because $\rho$ is not an edge of $T$, it cannot be realized as a
straight segment
in $\R^3$, because all of those vertex-vertex segments are edges of $T$.
However, $e$ is a straight segment in $\R^3$,
and so $|e| < |\rho|$, contradicting the assumption that $\rho$
is a shortest path.
\end{proof}

The reason this proof works for tetrahedra but can fail for a
polyhedron of $n > 4$ vertices is that the condition
that every pair of vertices are connected by an edge fails in general.
The polyhedra for which that holds are the ``neighborly polyhedra''
(more precisely, the 2-neighborly 3-polytopes).

\section{Hexagon Edges are Tetrahedron Edges}
\seclab{Edges}
With Lemma~\lemref{ShortestPaths} in hand, it only remains to show
that
the pairs of matched unit-length hexagon edges are shortest paths on
the manifold $M$ obtained by zipping the hexagon.
We use the labeling and folding shown in Figure~\figref{HexLabels}(a),
where $v_1$ and $v'_1$ are identified, as are $v_2$ and $v'_2$.
We need to show that both $v_0 v_1$ and $v_1 v_2$ are shortest paths
($v_2 v_3$ is symmetric to $v_0 v_1$ and follows by relabeling).

\begin{figure}[htbp]
\centering
\includegraphics[width=\linewidth]{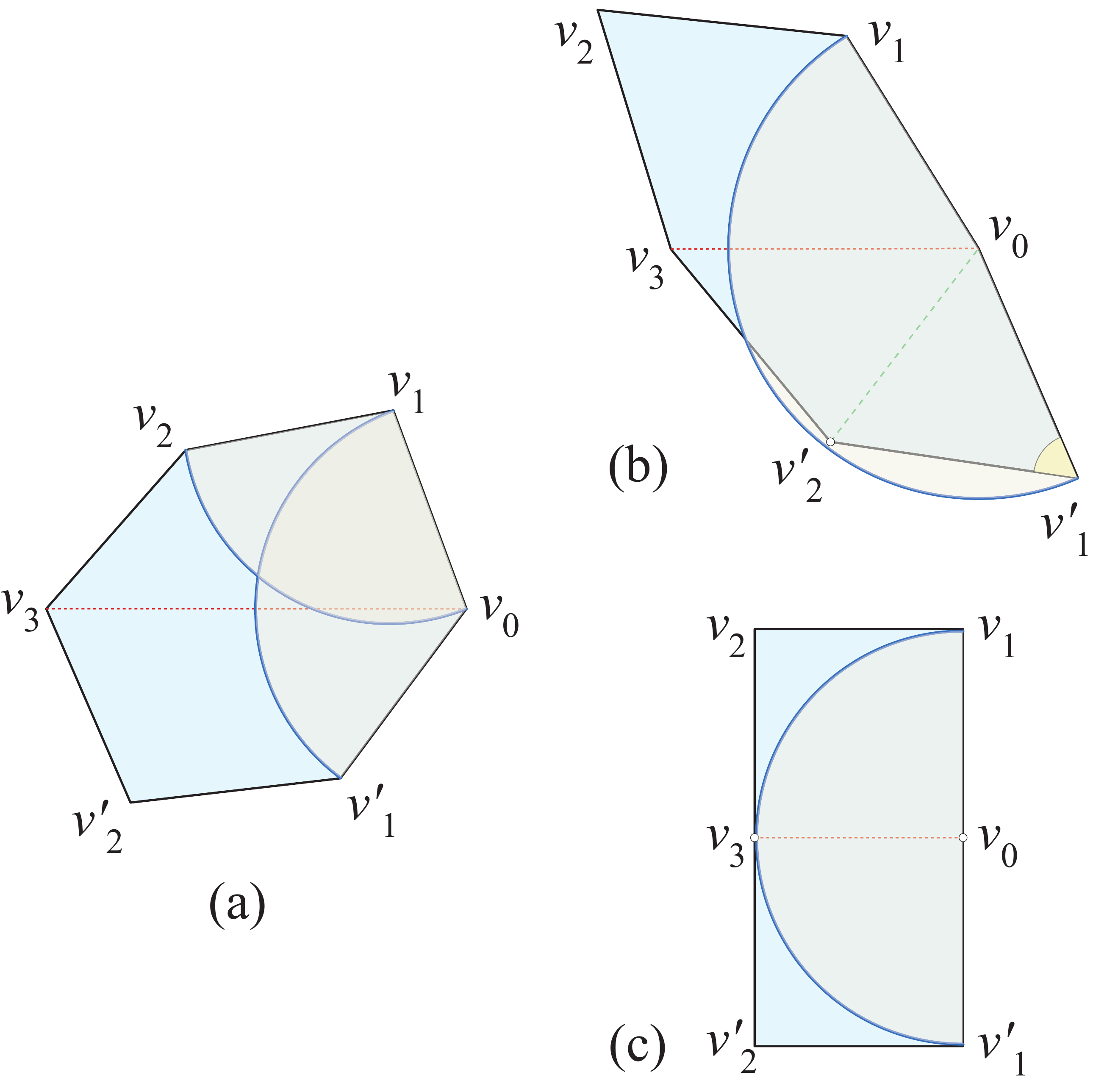}
\caption{(a)~Labels for the zipping of a hexagon.  The dashed edge shows
  the
perimeter halving line, but there is no assumption that diagonal is
a crease that becomes an edge of the tetrahedron.
(b)~When $D_0$ is not empty, some angle (here, $\a'_1$) is smaller than $\pi/3$.
(c)~A ``degenerate'' hexagon, with two angles $\pi$,
which folds to a doubly covered square.
}
\figlab{HexLabels}
\end{figure}

Let $D_i$ be the geodesic disk of unit
radius centered on $v_i$,  on the zipped manifold $M$.  If the disks $D_0$ and $D_1$, illustrated
in  Figure~\figref{HexLabels}(a), are empty of other vertices,
then the desired shortest paths are established.
It is here that we will employ the assumption that the hexagon angles are
greater than $\pi/3$.

First, for any unit-equilateral convex hexagon, without constraints on the
angles,
diagonals connecting opposite vertices are at least length~1:
$|v_0 v_3| \ge 1$,
$|v_1 v'_2| \ge 1$, and
$|v_2 v'_1| \ge 1$.
We now argue for this elementary fact 
(which is likely known in some guise
in the literature).

Concentrating on $D_0$ and $v_0 v_3$ (all others are equivalent by
relabeling),
consider the quadrilateral $(v_0,v_1,v_2,v_3)$.
Because the hexagon is unit-equilateral, 
$|v_0 v_1| = 1$,
$|v_1 v_2| = 1$, and
$|v_2 v_3| = 1$.
Assume for a contradiction that the diagonal is short: $|v_0 v_3| < 1$.
Then it is not difficult to prove that the sum of the quadrilateral
angles 
at the endpoints of this shorter side exceed $\pi$:
$\angle v_3 v_0 v_1 + \angle v_2 v_3 v_0 > \pi$.
Applying the same logic to the other half of the hexagon sharing
diagonal $v_0 v_3$,
the quadrilateral $(v_3, v'_2, v'_1, v_0)$,
we reach the conclusion that the sum of the hexagon angles
at $v_0$ and $v_3$ exceeds $2\pi$.
Thus one of those angles exceeds $\pi$, contradicting the fact
that the hexagon is convex.\footnote{
  I thank Mirela Damian for this argument, which is simpler than my
  original proof.}


We may conclude from this analysis is that $v_3$ cannot lie inside $D_0$.

Suppose now that $v'_2$ is inside $D_0$,
as in  Figure~\figref{HexLabels}(b).
Then, the angle $\a'_1$ at $v'_1$ must
be smaller than $\pi/3$, contradicting the angle minimum assumed in
the theorem.
The same clearly holds for $v_2$ inside $D_0$, as well as the other
combinations.
So the assumption that the hexagon is ``fat'' in the sense that none
of its angles are small, guarantees that the disks $D_0$ and $D_1$
(and by symmetry, $D_2$ and $D_3$) are empty of vertices of the
hexagon.

As is evident from Figure~\figref{HexLabels}(b), however,
even when no vertex is inside $D_0$, a portion of $D_0$ can fall
outside
the hexagon, which, because it is all zipped to a closed manifold
$M$, means it re-enters the hexagon at the other copy of that edge.
However, three facts are easily established.
First, the ``overhang'' has width at most $h=1 - \sqrt{3}/2$;
see Figure~\figref{DiskOverhang}.
Second, no vertex of the ``next'' hexagon copy can lie inside
that overhang (without violating convexity of the hexagon).
Third, to even have a hexagon edge be partially interior 
(and so overhanging the disk into a third hexagon copy)
requires some angle
($\b$ in the figure) to be very small,
violating the $\pi/3$ minimum angle.
Thus the overhang of a $D_i$ disk beyond the original
hexagon cannot encompass a vertex.
\begin{figure}[htbp]
\centering
\includegraphics[width=0.75\linewidth]{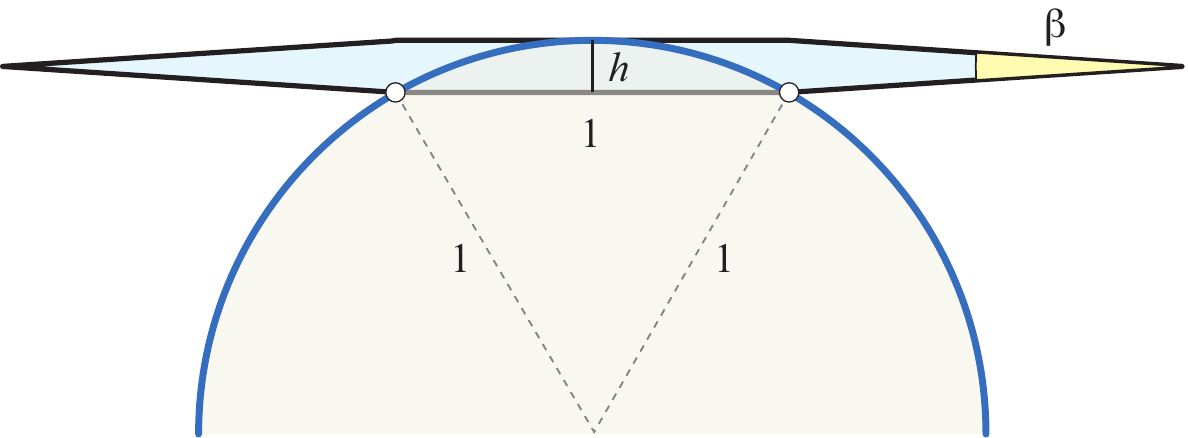}
\caption{A disk $D_i$ ``overhangs'' the original hexagon and enters
another copy.  Here $\b=2 \sin^{-1}( h/2 ) < 8 ^\circ$.}
\figlab{DiskOverhang}
\end{figure}

We may conclude:

\begin{lemma}
The three unit-length edges of the hexagon, 
$\{ v_0 v_1, v_1 v_2, v_2 v_3 \}$
are each shortest paths on the folded manifold $M$ between their endpoint vertices.
\lemlab{Edges}
\end{lemma}

I have firm empirical evidence 
(via explorations in Cinderella)
that Lemma~\lemref{Edges} holds
just as stated without any assumption that the hexagon is fat.
But proving this formally seems difficult.
Figure~\figref{Reflected01} illustrates a path $\rho$ that spirals
around the manifold from $v_0$ to $v_1$.
Each such geodesic path candidate must be established to be at least length 1.

\begin{figure}[htbp]
\centering
\includegraphics[width=\linewidth]{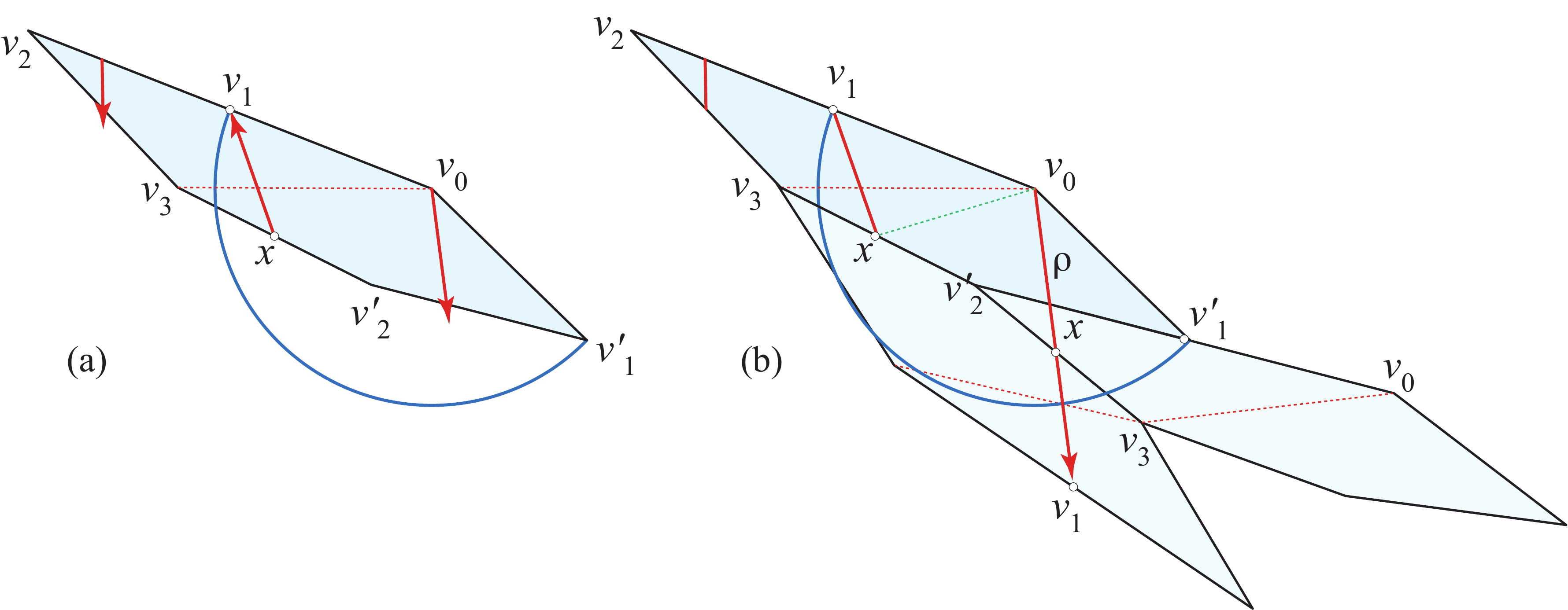}
\caption{(a)~A ``spiral'' path from $v_0$ to $v_1$.
(b)~The path flattened, crossing several copies of the hexagon.}
\figlab{Reflected01}
\end{figure}

\section{Discussion}
\seclab{Discussion}
We have now proved Theorem~\thmref{HexTheorem}:
Lemma~\lemref{Edges} shows that the three unit-length edges forming
a Hamiltonian path are shortest paths,
Lemma~\lemref{ShortestPaths} then implies that they are edges of the
tetrahedron,
and Lemma~\lemref{Distinct} establishes that the three tetrahedra are
distinct.

\begin{corollary}
There are an uncountable number of hexagons that satisfy the conditions of
Theorem~\thmref{HexTheorem}.
\corlab{Continuum}
\end{corollary}
\begin{proof}
The constraints that $\sum \a_i = 4 \pi$
and $\pi/3 < \a_i < \pi$
clearly leave an uncountable number of solutions, in fact a 5-dimensional open
set in $\R^5$
(the sixth angle is determined by the other five).
(For example, a small 5-ball around
$$
(\a_0,\a_1,\a_2,\a_3,\a_4) =
(\pi-\tfrac{1}{2},
\pi-\tfrac{1}{2},
\pi-\tfrac{1}{2},
\tfrac{\pi}{3}+\tfrac{1}{4},
\tfrac{\pi}{3}+\tfrac{1}{4})
$$
with $\a_5=\frac{\pi}{3}+1$, is inside this set.)
The constraint requiring independence over $\Q$ only excludes
a countable number of 4-dimensional hyperplanes
(e.g., $\a_2 = \pi - \frac{1}{2} \a_0$).
These hyperplanes have zero measure in $\R^5$,
and a countable union of sets of zero measure has zero measure,%
\footnote{
   I thank Qiaochu Yuan and Theo Buehler for guidance here,
   \url{http://math.stackexchange.com/questions/41494/}.
}
leaving an uncountable number of solutions after excluding
the hyperplanes.  
\end{proof}

The construction considered here generalizes to arbitary even $n$, although
I do not see how to prove that the zipper path follows
edges of the polyhedron:

\begin{proposition}
For any even $n \ge 4$, an equilateral convex $n$-gon is
the common zipper-unfolding of $n/2$ generally distinct polyhedra
of $(n-2)$ vertices each.
\end{proposition}

The regular-polygon version of this construction folds to
what were called ``pita polyhedra''
in~\cite[Sec.~25.7.2]{do-gfalop-07}.

I conjecture that all the zipper-paths in Proposition~1,
for strictly convex equilateral convex $n$-gons,
in fact follow polyhedron edges.
Resolving this conjecture would require tools to determine when
particular
geodesic paths on an Alexandrov-glued manifold are edges of
the resulting convex polyhedron.


\bibliographystyle{alpha}
\bibliography{/Users/orourke/bib/geom/geom}

\newcommand{\etalchar}[1]{$^{#1}$}
\begin{thebibliography}{LDD{\etalchar{+}}10}

\bibitem[DDLO00]{ddlo-ecerfu-00}
Erik~D. Demaine, Martin~L. Demaine, Anna Lubiw, and Joseph O'Rourke.
\newblock Examples, counterexamples, and enumeration results for foldings and
  unfoldings between polygons and polytopes.
\newblock Technical Report 069, Smith College, Northampton, July 2000.
\newblock arXiv:cs.CG/0007019.

\bibitem[DDLO02]{ddlo-efupp-02}
Erik~D. Demaine, Martin~L. Demaine, Anna Lubiw, and Joseph O'Rourke.
\newblock Enumerating foldings and unfoldings between polygons and polytopes.
\newblock {\em Graphs and Combin.}, 18(1):93--104, 2002.
\newblock See also~\cite{ddlo-ecerfu-00}.

\bibitem[DO07]{do-gfalop-07}
Erik~D. Demaine and Joseph O'Rourke.
\newblock {\em Geometric Folding Algorithms: Linkages, Origami, Polyhedra}.
\newblock Cambridge University Press, July 2007.
\newblock \url{http://www.gfalop.org}.

\bibitem[LDD{\etalchar{+}}10]{lddss-zupc-10}
Anna Lubiw, Erik Demaine, Martin Demaine, Arlo Shallit, and Jonah Shallit.
\newblock Zipper unfoldings of polyhedral complexes.
\newblock In {\em Proc. 22nd Canad. Conf. Comput. Geom.}, pages 219--222,
  August 2010.

\bibitem[O'R10]{o-fzupp-10}
Joseph O'Rourke.
\newblock Flat zipper-unfolding pairs for platonic solids.
\newblock \url{http://arxiv.org/abs/1010.2450}, October 2010.

\bibitem[She75]{s-cpcn-75}
Geoffrey~C. Shephard.
\newblock Convex polytopes with convex nets.
\newblock {\em Math. Proc. Camb. Phil. Soc.}, 78:389--403, 1975.

\end{thebibliography}
\end{document}